\documentclass[fleqn,12pt,twoside]{article}

\usepackage[headings]{espcrc1}
\readRCS
$Id: espcrc1.tex,v 1.2 2004/02/24 11:22:11 spepping Exp $
\usepackage{graphicx}
\usepackage[figuresright]{rotating}

\newtheorem{theorem}{Theorem}

\newenvironment{proof}[1][Proof.]{\begin{trivlist}
\item[\hskip \labelsep {\bfseries #1}]}{\end{trivlist}}

\newcommand{\AmS}{{\protect\the\textfont2
  A\kern-.1667em\lower.5ex\hbox{M}\kern-.125emS}}

\usepackage{amsmath}
\usepackage{amsfonts}
\title{A note on interval edge-colorings of graphs}

\author{R.R. Kamalian\address[MCSD]{Institute for Informatics and Automation Problems,\\
National Academy of Sciences, 0014, Armenia}%
\address{Department of Applied Mathematics and Informatics,\\
Russian-Armenian State University, 0051, Armenia}%
\thanks{email: rrkamalian@yahoo.com.},
        P.A. Petrosyan\addressmark[MCSD]%
\address{Department of Informatics and Applied Mathematics,\\
Yerevan State University, 0025, Armenia}%
\thanks {email: pet\_petros@\{ipia.sci.am, ysu.am, yahoo.com\}}}

\runtitle{A note on interval edge-colorings of
graphs}\runauthor{R.R. Kamalian, P.A. Petrosyan}

\begin{document}

\maketitle

\begin{abstract}
An edge-coloring of a graph $G$ with colors $1,2,\ldots ,t$ is
called an interval $t$-coloring if for each $i\in \{1,2,\ldots,t\}$
there is at least one edge of $G$ colored by $i$, and the colors of
edges incident to any vertex of $G$ are distinct and form an
interval of integers. In this paper we prove that if a connected
graph $G$ with $n$ vertices admits an interval $t$-coloring, then
$t\leq 2n-3$. We also show that if $G$ is a connected $r$-regular
graph with $n$ vertices has an interval $t$-coloring and $n\geq
2r+2$, then this upper bound can be improved to $2n-5$.\\


\end{abstract}

\bigskip

\section{Introduction}\

All graphs considered in this paper are finite, undirected, and have
no loops or multiple edges. Let $V(G)$ and $E(G)$ denote the sets of
vertices and edges of $G$, respectively. An $(a,b)$-biregular
bipartite graph $G$ is a bipartite graph $G$ with the vertices in
one part all having degree $a$ and the vertices in the other part
all having degree $b$. A partial edge-coloring of $G$ is a coloring
of some of the edges of $G$ such that no two adjacent edges receive
the same color. If $\alpha $ is a partial edge-coloring of $G$ and
$v\in V(G)$, then $S\left( v,\alpha \right)$ denotes the set of
colors of colored edges incident to $v$.

An edge-coloring of a graph $G$ with colors $1,2,\ldots ,t$ is
called an interval $t$-coloring if for each $i\in \{1,2,\ldots,t\}$
there is at least one edge of $G$ colored by $i$, and the colors of
edges incident to any vertex of $G$ are distinct and form an
interval of integers. A graph $G$ is interval colorable, if there is
$t\geq 1$ for which $G$ has an interval $t$-coloring. The set of all
interval colorable graphs is denoted by $\mathfrak{N}$. For a graph
$G\in \mathfrak{N}$, the greatest value of $t$ for which $G$ has an
interval $t$-coloring is denoted by $W\left(G\right)$.

The concept of interval edge-coloring was introduced by Asratian and
Kamalian \cite{b2}. In \cite{b2,b3} they proved the following
theorem.

\begin{theorem}
\label{mytheorem1} If $G$ is a connected triangle-free graph and
$G\in \mathfrak{N}$, then
\begin{center}
$W(G)\leq \vert V(G)\vert -1$.
\end{center}
\end{theorem}

In particular, from this result it follows that if $G$ is a
connected bipartite graph and $G\in \mathfrak{N}$, then $W(G)\leq
\vert V(G)\vert -1$. It is worth noting that for some families of
bipartite graphs this upper bound can be improved. For example, in
\cite{b1} Asratian and Casselgren proved the following

\begin{theorem}
\label{mytheorem2} If $G$ is a connected $(a,b)$-biregular bipartite
graph with $\vert V(G)\vert \geq 2(a+b)$ and $G\in \mathfrak{N}$,
then
\begin{center}
$W(G)\leq \vert V(G)\vert -3$.
\end{center}
\end{theorem}

For general graphs, Kamalian proved the following

\begin{theorem}
\label{mytheorem3}\cite{b6}. If $G$ is a connected graph and $G\in
\mathfrak{N}$, then
\begin{center}
$W(G)\leq 2\vert V(G)\vert -3$.
\end{center}
\end{theorem}

The upper bound of Theorem \ref{mytheorem3} was improved in
\cite{b5}.

\begin{theorem}
\label{mytheorem4}\cite{b5}. If $G$ is a connected graph with $\vert
V(G)\vert \geq 3$ and $G\in \mathfrak{N}$, then
\begin{center}
$W(G)\leq 2\vert V(G)\vert -4$.
\end{center}
\end{theorem}

On the other hand, in \cite{b7} Petrosyan proved the following
theorem.

\begin{theorem}
\label{mytheorem4} For any $\varepsilon >0$, there is a graph $G$
such that $G\in \mathfrak{N}$ and
\begin{center}
$W\left(G\right)\geq \left( 2-\varepsilon \right) \left\vert
V\left(G\right)\right\vert$.
\end{center}
\end{theorem}

For planar graphs, the upper bound of Theorem \ref{mytheorem3} was
improved in \cite{b4}.

\begin{theorem}
\label{mytheorem6}\cite{b4}. If $G$ is a connected planar graph and
$G\in \mathfrak{N}$, then
\begin{center}
$W(G)\leq \frac{11}{6}\vert V(G)\vert$.
\end{center}
\end{theorem}

In this note we give a short proof of Theorem \ref{mytheorem3} based
on Theorem \ref{mytheorem1}. We also derive a new upper bound for
the greatest possible number of colors in interval edge-colorings of
regular graphs.

\bigskip

\section{Main results}\

\begin{proof}[Proof of Theorem \ref{mytheorem3}.] Let $V(G)=\{v_{1},v_{2},\ldots,v_{n}\}$ and $\alpha$ be an
interval $W(G)$-coloring of the graph $G$. Define an auxiliary graph
$H$ as follows:
\begin{center}
$V(H)=U\cup W$, where
\end{center}
\begin{center}
$U=\{u_{1},u_{2},\ldots,u_{n}\}$, $W=\{w_{1},w_{2},\ldots,w_{n}\}$
and
\end{center}
\begin{center}
$E(H)=\left\{u_{i}w_{j},u_{j}w_{i}|~v_{i}v_{j}\in E(G), 1\leq i\leq
n,1\leq j\leq n\right\}\cup \{u_{i}w_{i}|~1\leq i\leq n\}$.
\end{center}

Clearly, $H$ is a connected bipartite graph with $\vert V(H)\vert =
2\vert V(G)\vert$.\\

Define an edge-coloring $\beta$ of the graph $H$ in the following
way:
\begin{description}
\item[(1)] $\beta (u_{i}w_{j})=\beta (u_{j}w_{i})=\alpha(v_{i}v_{j})+1$ for
every edge $v_{i}v_{j}\in E(G)$,

\item[(2)] $\beta (u_{i}w_{i})=\max S(v_{i},\alpha)+2$ for
$i=1,2,\ldots,n$.
\end{description}

It is easy to see that $\beta$ is an edge-coloring of the graph $H$
with colors $2,3,\ldots,W(G)+2$ and $\min S(u_{i},\beta)=\min
S(w_{i},\beta)$ for $i=1,2,\ldots,n$. Now we present an interval
$(W(G)+2)$-coloring of the graph $H$. For that we take one edge
$u_{i_{0}}w_{i_{0}}$ with $\min S(u_{i_{0}},\beta)=\min
S(w_{i_{0}},\beta)=2$, and recolor it with color $1$. Clearly, such
a coloring is an interval $(W(G)+2)$-coloring of the graph $H$.
Since $H$ is a connected bipartite graph and $H\in \mathfrak{N}$, by
Theorem \ref{mytheorem1}, we have
\begin{center}
$W(G)+2\leq \vert V(H)\vert -1 = 2\vert V(G)\vert-1$, thus
\end{center}
\begin{center}
$W(G)\leq 2\vert V(G)\vert-3$.
\end{center}
~$\square$
\end{proof}

\begin{theorem}
\label{mytheorem7} If $G$ is a connected $r$-regular graph with
$\vert V(G)\vert \geq 2r+2$ and $G\in \mathfrak{N}$, then
\begin{center}
$W(G)\leq 2\vert V(G)\vert -5$.
\end{center}
\end{theorem}
\begin{proof} In a similar way as in the prove of Theorem
\ref{mytheorem3}, we can construct an auxiliary graph $H$ and to
show that this graph has an interval $(W(G)+2)$-coloring. Next,
since $H$ is a connected $(r+1)$-regular bipartite graph with $\vert
V(H)\vert \geq 2(2r+2)$ and $H\in \mathfrak{N}$, by Theorem
\ref{mytheorem2}, we have
\begin{center}
$W(G)+2\leq \vert V(H)\vert -3 = 2\vert V(G)\vert-3$, thus
\end{center}
\begin{center}
$W(G)\leq 2\vert V(G)\vert-5$.
\end{center}
~$\square$
\end{proof}\

\end{document}